\documentclass[12pt,a4paper]{amsart}

\usepackage{amsmath,amsfonts,amssymb,mathtools,extarrows}
\usepackage{xcolor}
\usepackage{comment}

\usepackage{color}
\usepackage{graphicx}

\usepackage{enumerate}  

\usepackage{cancel}

\usepackage[hmargin=2cm,vmargin=2cm]{geometry}

\usepackage[hypertexnames=false,hyperfootnotes=false,colorlinks=true,linkcolor=blue,%
citecolor=purple,filecolor=magenta,urlcolor=cyan,unicode,linktocpage=true,pagebackref=false]{hyperref}
\usepackage{nameref,zref-xr}     

\usepackage{tikz}

\newcommand{\mbZ}{\mathbb Z}
\newcommand{\mbC}{\mathbb C}

\newcommand{\oM}{\overline{\mathcal M}}
\newcommand{\tg}{\widetilde g}
\newcommand{\tu}{{\widetilde u}}
\newcommand{\og}{\overline g}
\newcommand{\oh}{\overline h}
\newcommand{\hLambda}{\widehat\Lambda}

\def\oM{{\overline{\mathcal{M}}}}
\def\CP{{{\mathbb C}{\mathbb P}}}
\renewcommand{\Im}{\mathrm{Im}}

\def\d{{\partial}}

\newcommand{\<}{\left<}
\renewcommand{\>}{\right>}
\newcommand{\eps}{\varepsilon}

\newcommand{\str}{\mathrm{str}}

\newcommand{\cA}{\mathcal A}
\newcommand{\hcA}{\widehat{\mathcal A}}
\newcommand{\DR}{\mathrm{DR}}
\newcommand{\DZ}{\mathrm{DZ}}

\newcommand{\cF}{\mathcal F}

\newcommand{\Coef}{\mathrm{Coef}}

\DeclareMathOperator{\Deg}{Deg}

\newcommand{\tv}{\widetilde v}
\renewcommand{\top}{\mathrm{top}}

\newcommand{\gl}{\mathrm{gl}}

\newcommand{\sing}{\mathrm{sing}}
\newcommand{\hOmega}{\widehat{\Omega}}
\newcommand{\hE}{\widehat{E}}
\newcommand{\diag}{\mathrm{diag}}
\newcommand{\un}{{1\!\! 1}}
\newcommand{\mcF}{\mathcal{F}}

\newcommand{\of}{\overline{f}}

\newcommand{\tK}{\widetilde{K}}

\newcommand{\rspin}{{\text{$r$-spin}}}

\newcommand{\ov}{\overline{v}}

\newcommand{\wk}{\mathrm{wk}}

\newcommand{\bfu}{\mathbf{u}}

\newcommand{\Ker}{\mathrm{Ker}}
\newcommand{\e}{\mathrm{e}}
\newcommand{\rt}{\mathrm{rt}}
\newcommand{\wkt}{{\mathrm{wk},\mathrm{t}}}
\newcommand{\pol}{\mathrm{pol}}
\newcommand{\rtt}{{\mathrm{rt},\mathrm{t}}}
\newcommand{\lb}{\left(}
\newcommand{\rb}{\right)}
\newcommand{\cD}{\mathcal{D}}
\newcommand{\cS}{\mathcal{S}}
\newcommand{\ovr}{\overline{r}}
\newcommand{\tb}{\widetilde{b}}
\newcommand{\otau}{\overline{\tau}}

\newtheorem{theorem}{Theorem}[section]
\newtheorem{proposition}[theorem]{Proposition}
\newtheorem{lemma}[theorem]{Lemma}

\newtheorem{conjecture}[theorem]{Conjecture}

\theoremstyle{remark}

\newtheorem{remark}[theorem]{Remark}

\theoremstyle{definition}

\newtheorem{definition}[theorem]{Definition}

\usepackage{color}

\numberwithin{equation}{section}

\begin{document}

\title[Bihamiltonian structure of the DR hierarchy in the semisimple case]{Bihamiltonian structure of the DR hierarchy in the semisimple case}

\author{Alexandr Buryak}
\address{A. Buryak:\newline 
Faculty of Mathematics, National Research University Higher School of Economics, Usacheva str. 6, Moscow, 119048, Russian Federation;\smallskip\newline 
Skolkovo Institute of Science and Technology, Bolshoy Boulevard 30, bld.~1, Moscow, 121205, Russian Federation;\smallskip\newline
Sobolev Institute of Mathematics of the Siberian Branch of the Russian Academy of Sciences, Koptyug av.~4, Novosibirsk, 630090, Russian Federation}
\email{aburyak@hse.ru}

\author{Paolo Rossi}
\address{P. Rossi:\newline Dipartimento di Matematica ``Tullio Levi-Civita'', Universit\`a degli Studi di Padova,\newline
Via Trieste 63, 35121 Padova, Italy}
\email{paolo.rossi@math.unipd.it}

\begin{abstract}
Of the two approaches to integrable systems associated to semisimple cohomological field theories (CohFTs), the one suggested by Dubrovin and Zhang and the more recent one using the geometry of the double ramification (DR) cycle, the second has the advantage of being very explicit. The Poisson operator of the DR hierarchy is $\eta^{-1} \d_x$, where~$\eta$ is the metric of the CohFT, and the Hamiltonians are explicitly defined as generating functions of intersection numbers of the CohFT with the DR cycle, the top Hodge class $\lambda_g$, and powers of a psi-class. The question whether the DR hierarchy is endowed with a bihamiltonian structure appeared to be much harder. In our previous work in collaboration with S. Shadrin, when the CohFT is homogeneous, we proposed an explicit formula for a differential operator and conjectured that it would provide the required bihamiltonian structure. In this paper, we prove this conjecture. Our proof is based on two recently proved results: the equivalence of the DR hierarchy and the Dubrovin-Zhang hierarchy of a semisimple CohFT under Miura transformation and the polynomiality of the second Poisson bracket of the DZ hierarchy of a homogeneous semisimple CohFT. In particular, our second Poisson bracket coincides through the DR/DZ equivalence with the second Poisson bracket of the DZ hierarchy, hence providing a remarkably explicit approach to their bihamiltonian structure.
\end{abstract}

\date{\today}

\maketitle

\section{Introduction}

In \cite{DZ01} Dubrovin and Zhang present a construction associating an integrable hierarchy of Hamiltonian evolutionary PDEs in one space variable to a semisimple cohomological field theory (CohFT) on the moduli spaces of stable curves. This integrable system is called the Dubrovin-Zhang (DZ) hierarchy, it is tau-symmetric and (the logarithm of) the tau-function of one of its solutions coincides with the potential of the CohFT (the generating function of the intersection numbers of the CohFT with monomials in the psi-classes). One key fact is that the DZ hierarchy is, in an appropriate sense, nonsingular, meaning that Hamiltonians, Poisson structure and, consequently, the equations themselves are essentially polynomial in the spatial derivatives of the dependent variables. This polynomiality is not trivial a priori and was proved in \cite{BPS12}.\\

As explained in \cite{BPS12} the hierarchy can be defined starting from the genus $0$ part of the CohFT (as the principal hierarchy of the corresponding Dubrovin--Frobenius manifold). The principal hierarchy is the dispersionless part of the DZ hierarchy. The full DZ hierarchy is defined as the dispersionless part expressed in new coordinates via a coordinate transformation determined (in terms of the tau function) by higher genus geometric data. This coordinate transformation, however, is \emph{singular} in the spatial derivatives. Polynomiality for the full hierarchy, for the Hamiltonians and for its Poisson structure then corresponds to a series of highly nontrivial cancellations of the singular terms originating from applying the singular coordinate change to their regular dispersionless limit.\\

When the semisimple CohFT is homogeneous, Dubrovin defined in \cite{Dub96} a bihamiltonian structure for the principal hierarchy of the corresponding homogeneous Dubrovin--Frobenius manifold. Similar to what happens for the Hamiltonians and the first Poisson structure, polynomiality of the Poisson structure for the full DZ hierarchy obtained applying the singular coordinate change to Dubrovin's second Poisson bracket is a nontrivial fact which was only recently proved in \cite{LWZ21}.\\

In \cite{Bur15} the first named author introduced another construction of a Hamiltonian tau-symmetric integrable hierarchy, called double ramification (DR) hierarchy, from a CohFT, employing intersection theory of the CohFT with the DR cycle in the moduli space of stable curves. The construction is more general (it doesn't need semisimplicity and can be applied to partial CohFTs too), but in the semisimple case it was conjectured in \cite{Bur15,BDGR18} that the DR and DZ hierarchies coincide after a \emph{regular} (polynomial) coordinate transformation. This conjecture was transformed into a geometric statement about the validity of certain relations in the tautological ring of the moduli spaces of curves \cite{BDGR18,BDGR20,BGR19} and proved in \cite{BLS24}.\\

In \cite{BRS21}, in a collaboration with S. Shadrin, we conjectured a formula for a second Poisson structure for the DR hierarchy  of a homogeneous CohFT which would make it bihamiltonian.  This formula is an explicit and relatively simple closed expression in the generating function of intersection numbers of the DR cycle with the top Chern class of the Hodge bundle and the homogeneous CohFT. We could prove compatibility with the first Poisson bracket and gave some evidence of the bihamiltonian recursion for the hierarchy with respect to these two operators, but, in particular, we could not prove that our expression for the second bracket satisfies the Jacobi identity. Notice that, in light of the above mentioned equivalence between the DR and the DZ hierarchies, the regular coordinate transformation connecting the two hierarchies can be used to transport the bihamiltonian structure of the DZ hierarchy to the DR side, and our idea was that our second conjecturally Poisson operator would reproduce that bihamiltonian structure.\\

In this paper we prove this fact. The main idea is to study the singular coordinate transformation obtained by composing the singular coordinate transformation sending the dispersionless part of the DZ hierarchy (i.e. the principal hierarchy of the underlying Dubrovin--Frobenius manifold) to the full DZ hierarchy and the regular coordinate transformation connecting the DZ to the DR hierarchy. In Lemma \ref{lemma:singular-3} we prove that, in an appropriate sense, this composite transformation connecting the principal hierarchy to the DR hierarchy has no regular part (is purely singular in the language of this paper). In particular, when we apply it to Dubrovin's dispersionless second Poisson bracket of the principal hierarchy, this purely singular transformation produces mostly singular terms and only of handful of regular terms for the transformed bracket. At this point we make use of the above mentioned result of \cite{LWZ21} stating that the second Poisson bracket of the DZ hierarchy is polynomial, and so is its transformation through the (regular) DR/DZ equivalence proved in \cite{BLS24}. This means that all the above mentioned singular terms produced in the transformation have to vanish and only the handful of regular terms remain, which coincide with the formula we introduced in \cite{BRS21}.

\medskip

\subsection*{Notation and conventions}  

\begin{itemize}

\item We use the standard convention of sum over repeated Greek indices.

\smallskip

\item When it doesn't lead to a confusion, we use the symbol $*$ to indicate any value, in the appropriate range, of a sub- or superscript.

\smallskip

\item For a nonnegative integer $n$, let $[n]:=\{1,\dots,n\}$.

\smallskip

\item For a topological space $X$, we denote by $H^i(X)$ the cohomology groups with coefficients in $\mbC$. Let $H^{\e}(X):=\bigoplus_{i\ge 0}H^{2i}(X)$. By $\Deg\colon H^*(X)\to H^*(X)$ we denote the operator that acts on~$H^i(X)$ by the multiplication by $\frac{i}{2}$. 

\smallskip

\item We will work with the moduli spaces $\oM_{g,n}$ of stable algebraic curves of genus $g$ with $n$ marked points, which are nonempty only when the condition $2g-2+n>0$ is satisfied. We will often omit mentioning this condition explicitly, and silently assume that it is satisfied when a moduli space is considered. 

\end{itemize}

\medskip

\subsection*{Acknowledgements}

The work of A. B. is supported by the Mathematical Center in Akademgorodok under the agreement No. 075-15-2025-348 with the Ministry of Science and Higher Education of the Russian Federation. P.~R. is supported by the University of Padova and is affiliated to the INFN under the national project MMNLP and to the INdAM group GNSAGA.

\medskip


\section{DR hierarchy}

\begin{definition}[\cite{KM94}]
A \emph{cohomological field theory (CohFT)} is a family of linear maps 
$$
c_{g,n}\colon V^{\otimes n} \to H^\e(\oM_{g,n}),\quad g,n\ge 0,\quad 2g-2+n>0,
$$
where $V$ is an arbitrary finite dimensional vector space, together with a special element $e\in V$, called the \emph{unit}, and a symmetric nondegenerate bilinear form $\eta\colon V\times V\to\mbC$, called the \emph{metric}, such that the following axioms are satisfied:
\begin{itemize}
\item[(i)] The maps $c_{g,n}$ are equivariant with respect to the $S_n$-actions permuting the $n$ copies of~$V$ in $V^{\otimes n}$ and the $n$ marked points in $\oM_{g,n}$, respectively.

\smallskip

\item[(ii)] $\pi^*\left(c_{g,n}(\otimes_{i=1}^n v_i)\right) = c_{g,n+1}(\otimes_{i=1}^n  v_i\otimes e)$ for $v_1,\ldots,v_n\in V$, where $\pi\colon\oM_{g,n+1}\to\oM_{g,n}$ is the map that forgets the last marked point. Moreover, $c_{0,3}(v_1\otimes v_2 \otimes e) =\eta(v_1,v_2)$ for $v_1,v_2\in V$, where we use the identification $H^*(\oM_{0,3})=\mbC$ coming from the fact that~$\oM_{0,3}$ is a point.

\smallskip

\item[(iii)] Choosing a basis $e_1,\ldots,e_{\dim V}$ of $V$, we have
$$
\gl^*\left(c_{g_1+g_2,n_1+n_2}( \otimes_{i=1}^{n_1+n_2} e_{\alpha_i})\right) = \eta^{\mu \nu}c_{g_1,n_1+1}(\otimes_{i\in I} e_{\alpha_i} \otimes e_\mu)\otimes c_{g_2,n_2+1}(\otimes_{j\in J} e_{\alpha_j}\otimes e_\nu)
$$
for $1\leq\alpha_1,\ldots,\alpha_{n_1+n_2}\leq \dim V$, where $I \sqcup J =[n_1+n_2]$, $|I|=n_1$, $|J|=n_2$, $\gl\colon\oM_{g_1,n_1+1}\times\oM_{g_2,n_2+1}\to \oM_{g_1+g_2,n_1+n_2}$ is the corresponding gluing map, $\eta_{\alpha\beta}:=\eta(e_\alpha,e_\beta)$, we will also denote by $\eta$ the matrix $(\eta_{\alpha\beta})$, and~$\eta^{\alpha\beta}$ is defined by $(\eta^{\alpha\beta}):=\eta^{-1}$. Clearly, the axiom doesn't depend on the choice of a basis in $V$.

\smallskip

\item[(iv)] $\gl^*\left(c_{g+1,n}(\otimes_{i=1}^n e_{\alpha_i})\right) = c_{g,n+2}(\otimes_{i=1}^n e_{\alpha_i}\otimes e_{\mu}\otimes e_\nu) \eta^{\mu \nu}$ for $1 \leq\alpha_1,\ldots,\alpha_n\leq \dim V$, where  $\gl\colon\oM_{g,n+2}\to \oM_{g+1,n}$ is the gluing map that increases the genus by identifying the last two marked points.
\end{itemize}
\end{definition}

\medskip

\begin{definition} 
A CohFT $\{c_{g,n}\}$ is called \emph{homogeneous} if the vector space $V$ is graded, $V=\bigoplus_{q\in\mbC} V_q$ (only finitely many vector space $V_q$ are nonzero), with $\deg e=0$, and there exist a vector $\overline{r}\in V$ and a complex number $\delta$ such that for any homogeneous vectors $v_1,\ldots,v_n\in V$ the following condition is satisfied:
\begin{gather*}
\Deg c_{g,n}(\otimes_{i=1}^n v_i)+\pi_*\left(c_{g,n+1}(\otimes_{i=1}^n v_i\otimes \overline{r})\right)=\left(\sum_{i=1}^n \deg v_i+\delta(g-1)\right)c_{g,n}(\otimes_{i=1}^n v_i),
\end{gather*}
where $\pi\colon\oM_{g,n+1}\to\oM_{g,n}$ is the map that forgets the last marked point.
\end{definition}

\medskip

Let us briefly recall basic notions and notations in the formal theory of evolutionary PDEs with one spatial variable:
\begin{itemize}
\item We fix an integer $N\ge 1$ and consider formal variables $u^1,\ldots,u^N$. To the formal variables $u^\alpha$ we attach formal variables $u^\alpha_d$ with $d\ge 0$ and introduce the algebra of \emph{differential polynomials} $\cA_u:=\mbC[[u^*_0]][u^*_{\ge 1}]$. We identify $u^\alpha_0=u^\alpha$ and also denote $u^\alpha_x:=u^\alpha_1$, $u^{\alpha}_{xx}:=u^\alpha_2$, \ldots. 

\smallskip

\item Operators $\d_x,\frac{\delta}{\delta u^\alpha}\colon\cA_u\to\cA_u$ are defined by $\d_x:=\sum_{d\ge 0}u^\alpha_{d+1}\frac{\d}{\d u^\alpha_d}$ and $\frac{\delta}{\delta u^\alpha}:=\sum_{d\ge 0}(-\d_x)^d\circ\frac{\d}{\d u^\alpha_d}$, respectively.

\smallskip

\item The space $\Lambda_u:=\left.\cA_u\right/(\mbC\oplus\Im(\d_x))$ is called the space of \emph{local functionals}. The image of $f\in\cA_u$ under the canonical projection $\cA_u\to\Lambda_u$ is denoted by $\int f dx\in\Lambda_u$. Since $\mbC\oplus\Im(\d_x)\subset\Ker\left(\frac{\delta}{\delta u^\alpha}\right)$, there is a well-defined linear map $\frac{\delta}{\delta u^\alpha}\colon\Lambda_u\to\cA_u$.

\smallskip

\item Denote by $\cA_{u;d}\subset\cA_u$ and $\Lambda_{u;d}\subset\Lambda_u$ the homogeneous components of (differential) degree $d$, where $\deg u^\alpha_i:=i$.

\smallskip

\item The extended spaces of differential polynomials and local functionals are defined by $\hcA_u:= \cA_u[[\eps]]$ and $\hLambda_u:=\Lambda_u[[\eps]]$, respectively. Let $\hcA_{u;d}\subset\hcA_u$ and $\hLambda_{u;d}\subset\hLambda_u$ be the homogeneous components of degree~$d$, where $\deg\eps:=-1$.  

\smallskip

\item A \emph{differential operator} on $\hcA_u$ is an operator of the form $\sum_{i,j\ge 0}f^{[i]}_j\eps^i\d_x^j$, where $f^{[i]}_j\in\cA_u$ and for any fixed $i$ the sum $\sum_{j\ge 0}f^{[i]}_j\d_x^j$ has only finitely many nonzero terms. For any differential operator $K=\sum_{i\ge 0} f_i\d_x^i$, we define $K^\dagger:=\sum_{i\ge 0}(-\d_x)^i\circ f_i$.

\smallskip

\item A \emph{matrix differential operator} is an $N\times N$ matrix~$K=(K^{\mu\nu})$ consisting of differential operators~$K^{\mu\nu}$. We define $K^\dagger=\lb(K^\dagger)^{\mu\nu}\rb$ by $(K^\dagger)^{\mu\nu}:=\lb K^{\nu\mu}\rb^\dagger$.

\smallskip

\item Given a matrix differential operator~$K=(K^{\mu\nu})$, a bracket on the space $\hLambda_u$ is defined by $\{\overline{f},\overline{g}\}_{K}:=\int\left(\frac{\delta \overline{f}}{\delta u^\mu}K^{\mu \nu}\frac{\delta \overline{g}}{\delta u^\nu}\right)dx$, where $\of,\og\in\hLambda_u$. A \emph{Poisson operator} is a matrix differential operator $K$ such that the bracket $\{\cdot,\cdot\}_K$ is skewsymmetric and satisfies the Jacobi identity. The space of Poisson operators will be denoted by $\mathcal{PO}_u$. For any symmetric matrix $M=(m^{\alpha\beta})$ with complex coefficients, the matrix differential operator~$M\d_x$ is Poisson.

\smallskip

\item A {\it Hamiltonian system} of PDEs is a system of the form
\begin{gather*}
\frac{\partial u^\alpha}{\partial \tau_i} = K^{\alpha\mu} \frac{\delta\overline{h}_i}{\delta u^\mu},\quad \alpha\in[N],\quad i\in\mbZ_{\ge 1},
\end{gather*}
where $K$ is a Poisson operator and $\oh_i\in\hLambda_u$. The local functionals $\oh_i$ are called the {\it Hamiltonians}. The system is called {\it integrable} if $\{\oh_i,\oh_j\}_K=0$ for all $i,j\geq 1$.

\smallskip

\item We will discuss solutions of Hamiltonian systems. Note that for any differential polynomial $P\in\hcA_u$ and an $N$-tuple $(\bfu^1,\ldots,\bfu^N)\in\mbC[[x,\tau_*,\eps]]^N$ satisfying $\bfu^\alpha|_{x=\tau_*=\eps=0}=0$, the substitution $P|_{u^\alpha_n\mapsto\d_x^n\bfu^\alpha}$ is well defined as a formal power series in the variables~$x$,~$\tau_i$, and~$\eps$. So we will consider solutions $(\bfu^1,\ldots,\bfu^N)$ of Hamiltonian systems only in the class of formal power series satisfying the condition $\bfu^\alpha|_{x=\tau_*=\eps=0}=0$.

\smallskip

\item Consider an integrable Hamiltonian system. Given an arbitrary $N$-tuple $(f^1,\ldots,f^N)\in\mbC[[x,\eps]]^N$ satisfying $f^\alpha|_{x=\eps=0}=0$, the system has a unique solution $(\bfu^1,\ldots,\bfu^N)\in\mbC[[x,\tau_*,\eps]]$ satisfying the initial condition $\bfu^\alpha|_{\tau_*=0}=f^\alpha$.

\smallskip

\item A \emph{Miura transformation} is a change of variables of the form
\begin{align*}
&u^\alpha\mapsto w^\alpha(u^*_*,\eps)=\sum_{k\ge 0}\eps^k f^\alpha_k(u^*_*),\quad \alpha\in[N],\\
&f^\alpha_k\in\cA_{u;k},\quad f^\alpha_0|_{u^*=0}=0,\quad\left.\det\left(\frac{\d f_0^\alpha}{\d u^\beta}\right)\right|_{u^*=0}\ne 0.
\end{align*}
Using this change of variables, any differential polynomial $f\in\hcA_{u}$ can be written as a differential polynomial in the variables $w^1,\ldots,w^N$, which we will denote by~$f(w^*_*,\eps)$. Therefore, a Miura transformation gives an isomorphism $\hcA_u\cong\hcA_w$.

\smallskip

\item A Miura tranformation gives an isomorphism $\hLambda_u\cong\hLambda_w$. A local functional $\oh\in\hLambda_u$ written in the variables~$w^1,\ldots,w^N$ will be denoted by~$\oh[w]$.

\smallskip

\item Consider an arbitrary Miura transformation and the associated isomorphism $\hLambda_u\cong\hLambda_w$. Any matrix differential operator $K$ gives a bracket $\{\cdot,\cdot\}_K$ on $\hLambda_u$, which, under the isomorphism $\hLambda_u\cong\hLambda_w$, induces a bracket $\{\cdot,\cdot\}_{K_w}$ on $\hLambda_w$ where $K_w=(K_w^{\alpha\beta})$ is the matrix differential operator given by
\begin{gather*}
K_{w}^{\alpha\beta}:=\left.\lb\sum_{p,q\ge 0}\frac{\d w^\alpha(u^*_*,\eps)}{\d u^\mu_p}\d_x^p\circ K^{\mu\nu}\circ(-\d_x)^q\circ\frac{\d w^\beta(u^*_*,\eps)}{\d u^\nu_q}\rb\right|_{u^\gamma_m\mapsto \d_x^m u^\gamma(w^*_*,\eps)}.
\end{gather*}
\end{itemize}

\medskip

\begin{remark}\label{remark:indentifying x and t}
We will mainly consider Hamiltonian systems where one of the flows, say $\frac{\d}{\d\tau_1}$, is given by $\frac{\d u^\alpha}{\d\tau_1}=u^\alpha_x$. Then any solution of such system has the form $(\bfu^1,\ldots,\bfu^N)|_{\tau_1\mapsto\tau_1+x}$ for some $\bfu^\alpha\in\mbC[[\tau_*,\eps]]$. Abusing notation, we will call the $N$-tuple $(\bfu^1,\ldots,\bfu^N)\in\mbC[[\tau_*,\eps]]^N$ a solution of the system.
\end{remark}

\medskip

We will use the following standard cohomology classes on $\oM_{g,n}$:
\begin{itemize}
\item The psi-class $\psi_i\in H^2(\oM_{g,n})$, $1\le i\le n$, is the first Chern class of the line bundle over~$\oM_{g,n}$ formed by the cotangent lines at the $i$-th marked point of stable curves.

\smallskip

\item The Hodge class $\lambda_j:= c_j(\mathbb E)\in H^{2j}(\oM_{g,n})$, $j\ge 0$, where~$\mathbb E$ is the rank~$g$ Hodge vector bundle over~$\oM_{g,n}$ whose fibers are the spaces of holomorphic one-forms on stable curves. 

\smallskip

\item The \emph{double ramification (DR) cycle} $\DR_g(a_1,\ldots,a_n)\in H^{2g}(\oM_{g,n})$, $a_1,\ldots,a_n\in\mbZ$, $\sum a_i=0$, is defined as follows. There is a moduli space of projectivized stable maps to~$\CP^1$ relative to~$0$ and~$\infty$, with ramification profile over $0$ given by the negative numbers among the $a_i$-s, ramification profile over $\infty$ given by the positive numbers among the $a_i$-s, and the zeros among the $a_i$-s correspond to additional marked points (see, e.g.,~\cite{BSSZ15} for more details). This moduli space is endowed with a virtual fundamental class, which lies in the homology group of degree $2(2g-3+n)$. The DR cycle~$\DR_g(a_1,\ldots,a_n)$ is the Poincar\'e dual to the pushforward, through the forgetful map to~$\oM_{g,n}$, of this virtual fundamental class. The crucial property of the DR cycle is that for any cohomology class $\theta\in H^*(\oM_{g,n})$ the integral $\int_{\oM_{g,n+1}}\lambda_g\DR_g\left(-\sum a_i,a_1,\ldots,a_n\right)\theta$ is a homogeneous polynomial in $a_1,\ldots,a_n$ of degree~$2g$ (see, e.g.,~\cite[Lemma~3.2]{Bur15}). 
\end{itemize}

\medskip

For a given CohFT $\left\{c_{g,n}\colon V^{\otimes n} \to H^\e(\oM_{g,n})\right\}$, define differential polynomials $g_{\alpha,d}\in\hcA_{u;0}$, $\alpha\in[N]$, $d\in\mbZ_{\ge 0}$, as follows:
\begin{multline*}
g_{\alpha,d}:=\sum_{\substack{g,n\ge 0\\2g-1+n>0}}\frac{\eps^{2g}}{n!}\sum_{\substack{b_1,\ldots,b_n\ge 0\\b_1+\ldots+b_n=2g}}u^{\alpha_1}_{b_1}\cdots u^{\alpha_n}_{b_n}\times\\
\times\Coef_{a_1^{b_1}\cdots a_n^{b_n}}\left(\int_{\oM_{g,n+1}}\DR_g\left(-\sum a_i,a_1,\ldots,a_n\right)\lambda_g\psi_1^d c_{g,n+1}(e_\alpha\otimes\otimes_{i=1}^n e_{\alpha_i})\right).
\end{multline*}
Define $\og_{\alpha,d}:=\int g_{\alpha,d}dx$.

\medskip

\begin{definition}[\cite{Bur15}]
The Hamiltonian system
$$
\frac{\d u^\alpha}{\d t^\beta_q}=\eta^{\alpha\mu}\d_x\frac{\delta\og_{\beta,q}}{\delta u^\mu},\quad 1\le\alpha,\beta\le N,\quad q\in\mbZ_{\ge 0},
$$
is called the \emph{DR hierarchy}.
\end{definition}

\medskip

The Hamiltonians~$\og_{\alpha,d}$ satisfy $\{\og_{\alpha_1,d_1},\og_{\alpha_2,d_2}\}_{\eta^{-1}\d_x}=0$~\cite[Theorem~4.1]{Bur15}, so the DR hierarchy is integrable.

\medskip

We equip the DR hierarchy with the following~$N$ linearly independent Casimirs of the Poisson bracket $\{\cdot,\cdot\}_{\eta^{-1}\d_x}$:
$$
\og_{\alpha,-1}:=\int\eta_{\alpha\beta}u^\beta dx,\quad \alpha\in[N].
$$
Another important object related to the DR hierarchy is the local functional $\og\in\hLambda_{u;0}$ defined by the relation~\cite[Section 4.2.5]{Bur15}:
\begin{gather*}
\og_{\un,1}=(D-2)\og,
\end{gather*}
where $D:=\sum_{n\ge 0}(n+1)u^\alpha_n\frac{\d}{\d u^\alpha_n}$, $\og_{\un,1}:=A^\alpha\og_{\alpha,1}$, and the complex numbers $A^\alpha$ are defined by $e=A^\alpha e_\alpha$.

\medskip


\section{Bihamiltonian structure}

We associate with a differential polynomial $f\in\hcA_u$ a sequence of differential operators indexed by $\alpha\in[N]$ and $k\in\mbZ_{\ge 0}$: 
\begin{gather*}
L_\alpha^k(f):=\sum_{i\ge k}{i\choose k}\frac{\d f}{\d u^\alpha_i}\d_x^{i-k}.
\end{gather*}
We will often use the notation $L_\alpha(f):=L^0_\alpha(f)$. These operators satisfy the property
\begin{gather*}
L_\alpha^k(\d_x f)=\d_x\circ L_\alpha^k(f)+L_\alpha^{k-1}(f),\quad k\ge 0,
\end{gather*}
where we adopt the convention $L_{\alpha}^{l}(f):=0$ for $l<0$. Given a symmetric matrix $(\eta^{\alpha\beta})$ with complex coefficients, we associate with any local functional $\oh\in\hLambda_u$ a sequence of matrix differential operators $\hOmega^k(\oh)=\left(\hOmega^k(\oh)^{\alpha\beta}\right)$, indexed by $k\in\mbZ_{\ge 0}$, defined by
$$
\hOmega^k(\oh)^{\alpha\beta}:=\eta^{\alpha\mu}\eta^{\beta\nu}L^k_\nu\left(\frac{\delta\oh}{\delta u^\mu}\right).
$$
These operators satisfy the property $\hOmega^k(\og)^\dagger=(-1)^k\hOmega^k(\og)$ \cite[Lemma~1.5]{BRS21}. We will often use the notation $\hOmega(\oh):=\hOmega^0(\oh)$.

\medskip

Consider a homogeneous CohFT $\{c_{g,n}\colon V^{\otimes n}\to H^\e(\oM_{g,n})\}$. We fix a homogeneous basis $e_1,\ldots,e_N\in V$, with $\deg e_\alpha=q_\alpha$, and introduce complex numbers $r^\gamma$ by $\overline{r}=r^\gamma e_\gamma$. Define a diagonal matrix $\mu$ by
$$
\mu:=\diag(\mu_1,\ldots,\mu_N),\quad\text{where}\quad \mu_\alpha:=q_\alpha-\frac{\delta}{2}.
$$
Note that $\mu\eta+\eta\mu=0$. Define an operator $\hE$ on $\hcA_u$ by
$$
\hE:=\sum_{n\ge 0}\left((1-q_\alpha)u^\alpha_n+\delta_{n,0}r^\alpha\right)\frac{\d}{\d u^\alpha_n}+\frac{1-\delta}{2}\eps\frac{\d}{\d\eps}.
$$
Consider now the DR hierarchy associated to our homogeneous CohFT.

\medskip

\begin{definition}[\cite{BRS21}]
Define a matrix differential operator $K^\DR=(K^{\DR;\alpha\beta})$ by
\begin{gather}\label{eq:formula for KDR}
K^\DR:=\hE\left(\hOmega(\og)\right)\circ\d_x+\hOmega(\og)_x\circ\left(\frac{1}{2}-\mu\right)+\d_x\circ\hOmega^1(\og)\circ\d_x,
\end{gather}
where the notation $\hE\big(\hOmega(\og)\big)$ (respectively, $\hOmega(\og)_x$) means that we apply the operator $\hE$ (respectively, $\d_x$) to the coefficients of the operator $\hOmega(\og)$.
\end{definition}

\medskip

Recall that two Poisson operators $K_1$ and $K_2$ are said to be \emph{compatible} if the linear combination $K_2-\lambda K_1$ is a Poisson operator for any $\lambda\in\mbC$.

\medskip

\begin{conjecture}[\cite{BRS21}]{\ }
\label{conjecture:main}
\begin{enumerate}
\item The matrix differential operator $K^\DR$ is Poisson and is compatible with the Poisson operator $\eta^{-1}\d_x$.

\smallskip

\item The Poisson operator $K^\DR$ endows the DR hierarchy with a bihamiltonian structure with the following bihamiltonian recursion relation:
\begin{gather}\label{eq:bihamiltonian recursion for DR hierarchy}
\left\{\cdot,\og_{\alpha,d}\right\}_{K^\DR}=\left(d+\frac{3}{2}+\mu_\alpha\right)\left\{\cdot,\og_{\alpha,d+1}\right\}_{\eta^{-1}\d_x}+A^\beta_\alpha\left\{\cdot,\og_{\beta,d}\right\}_{\eta^{-1}\d_x},\quad d\ge -1,
\end{gather}
where $A^\beta_\alpha:=\eta^{\beta\nu}c_{0,3}(e_\nu\otimes e_\alpha\otimes\overline{r})$.
\end{enumerate}
\end{conjecture}

\medskip

\begin{remark}\label{remark:more on KDR}{\ }
\begin{itemize}
\item Consider the Dubrovin--Frobenius potential 
$$
F:=\sum_{n\ge 3}\sum_{\alpha_1,\ldots,\alpha_n\in[N]}\frac{\prod t^{\alpha_i}}{n!}\int_{\oM_{0,n}}c_{0,n}(\otimes_{i=1}^n e_{\alpha_i})\in\mbC[[t^1,\ldots,t^N]]
$$
corresponding to our CohFT. Denote 
$$
\tb^{\alpha\beta}:=\left.\eta^{\alpha\nu}\eta^{\beta\theta}\frac{\d^2F}{\d t^\nu\d t^\theta}\right|_{t^\gamma\mapsto u^\gamma},\qquad \tg^{\alpha\beta}:=((1-q_\gamma)u^\gamma+r^\gamma)\frac{\d\tb^{\alpha\beta}}{\d u^\gamma}.
$$
Then we have
\begin{gather*}
\left.K^{\DR;\alpha\beta}\right|_{\eps=0}=\tg^{\alpha\beta}\d_x+\d_x \tb^{\alpha\beta}\left(\frac{1}{2}-\mu_\beta\right),
\end{gather*}
and we see that the right-hand side coincides with a famous Poisson operator found by Dubrovin~\cite[Theorem~6.3]{Dub96}, which endows the principal hierarchy associated with $F$ with a bihamiltonian structure (see more details in Section~\ref{section:DZ hierarchy} below). So we see that the operator $K^\DR$ can be naturally considered as a deformation of the Dubrovin operator. We also see that the local functional~$\og$ in formula~\eqref{eq:formula for KDR} plays the role of the potential~$F$ in Dubrovin's formula.  

\smallskip

\item In~\cite[Lemma~1.14]{BRS21}, the authors proved the following alternative expression for the operator $K^\DR$:
\begin{gather*}
K^\DR=\d_x\circ\hOmega(\og)\circ\left(\frac{1}{2}-\mu\right)+\left(\frac{1}{2}-\mu\right)\circ\hOmega(\og)\circ\d_x+\eta^{-1}A\eta^{-1}\d_x+\d_x\circ\hOmega^1(\og)\circ\d_x,
\end{gather*}
where the matrix $A=(A_{\alpha\beta})$ is given by $A_{\alpha\beta}:=c_{0,3}(e_\alpha\otimes e_\beta\otimes\ovr)$.
\end{itemize}
\end{remark}

\medskip

Recall that a CohFT is called \emph{semisimple}, if the algebra defined by the structure constants $\theta^\alpha_{\beta\gamma}:=\eta^{\alpha\nu}c_{0,3}(e_\nu\otimes e_\beta\otimes e_\gamma)$ doesn't have nilpotents.

\medskip

\begin{theorem}\label{theorem:main}
Conjecture~\ref{conjecture:main} is true if the homogeneous CohFT is semisimple.
\end{theorem}

After some preparation in the next section, we will prove this theorem in Section~\ref{section:proof of the main theorem}.

\medskip

\begin{remark}
There is a simple argument allowing to extend the class of CohFTs, beyond semisimple CohFTs, for which we can prove Conjecture~\ref{conjecture:main}. Consider an arbitrary homogeneous CohFT $\{c_{g,n}\colon V^{\otimes n}\to H^\e(\oM_{g,n})\}$, with fixed homogeneous basis $e_1,\ldots,e_N\in V$. Let $\tau^1,\ldots,\tau^N$ be formal variables and $\otau=(\tau^1,\ldots,\tau^N)$. A \emph{formal shift} of this CohFT is a CohFT with coefficients in $\mbC[[\tau^1,\ldots,\tau^N]]$
$$
\left\{c^{\otau}_{g,n}\colon V^{\otimes n}\to H^\e(\oM_{g,n})\otimes\mbC[[\tau^1,\ldots,\tau^N]]\right\}
$$
defined by (see, e.g,~\cite[Section~1]{PPZ15})
$$
c_{g,n}^{\otau}(\otimes_{i=1}^n v_i):=\sum_{m\ge 0}\frac{1}{m!}(\pi_m)_*\lb c_{g,n+m}(\otimes_{i=1}^n v_i\otimes (\tau^\alpha e_\alpha)^{\otimes m})\rb\in H^\e(\oM_{g,n})\otimes\mbC[[\tau^1,\ldots,\tau^N]],
$$
where $\pi_m\colon\oM_{g,n+m}\to\oM_{g,n}$ is the map that forgets the last $m$ marked points. Suppose that there exists a~neighborhood $U$ of~$0$ in~$V$ such that $c_{g,n}^{\otau}(\otimes_{i=1}^n v_i)$ is convergent for any~$\tau^\alpha e_\alpha\in U$. Then, for any such vector $\otau\in\mbC^N$, the CohFT $\{c_{g,n}^{\otau}\}$ is homogeneous with the same constants~$q_\alpha$ and~$\delta$, but with the constant~$r^\alpha$ replaced by $r^\alpha+\tau^\alpha(1-q_\alpha)$. The coefficients of the differential polynomials $g_{\alpha,d}$ and of the operator~$K^\DR$, constructed for the CohFT $\{c_{g,n}^{\otau}\}$, depend holomorphically on $\otau$. Suppose moreover, that the CohFT $\{c_{g,n}^{\otau}\}$ is semisimple for generic $\tau^\alpha e_\alpha\in U$. Then, taking the limit at nonsemisimple points, using Theorem~\ref{theorem:main}, we conclude that Conjecture~\ref{conjecture:main} is true for any CohFT $\{c_{g,n}^{\otau}\}$ with $\tau^\alpha e_\alpha \in U$. An important example is given by the $r$-spin CohFT $\{c_{g,n}^{\rspin}\colon(\mbC^{r-1})^{\otimes n}\to H^\e(\oM_{g,n})\}$ (see, e.g,~\cite[Section~1]{PPZ15}). This CohFT is not semisimple for any $r\ge 3$, but the formal shift is convergent for any $\otau\in\mbC^{r-1}$ and is semisimple for generic $\otau\in\mbC^{r-1}$. We conclude that Conjecture~\ref{conjecture:main} is true for the $r$-spin CohFT.
\end{remark}

\medskip


\section{Rational Miura transformations}

We fix a positive integer $N$ throughout this section.

\medskip

\subsection{Definitions}

Consider formal variables $v^1,\ldots,v^N$. Define an extension~$\cA^\wk_v$ of the algebra~$\cA_v$ as the algebra of formal power series in the shifted variables $(v^\alpha_n-\delta^{\alpha,1}\delta_{n,1})$, $\alpha\in[N]$, $n\in\mbZ_{\ge 0}$, with complex coefficients. Denote by $\cA^\wkt_v\subset\cA^\wk_v$ the subalgebra consisting of elements $f\in\cA^\wk_v$ satisfying $\frac{\d f}{\d v^\alpha_n}=0$ for all $\alpha\in[N]$ and big enough~$n$. Denote $\hcA^\wk_v:=\cA^\wk_v[[\eps]]$ and $\hcA^\wkt_v:=\cA^\wkt_v[[\eps]]$. We have the chain of inclusions
$$
\hcA_{v}\subset\hcA^\wkt_{v}\subset\hcA^\wk_{v}.
$$
Note that the operator $\d_x\colon\hcA_v\to\hcA_v$ can be extended to the algebra $\hcA^\wk_v$ by the same formula: 
$$
\d_x:=\sum_{i\ge 0}v^\alpha_{i+1}\frac{\d}{\d v^\alpha_i}\colon\hcA^\wk_v\to\hcA^\wk_v.
$$
For $f\in\hcA^\wk_v$ and $d\in\mbZ$, we will write $\deg f=d$ if $\lb\sum_{n\ge 0}n\,v^\alpha_n\frac{\d}{\d v^\alpha_n}-\eps\frac{\d}{\d\eps}\rb f=d\cdot f$. We denote by $\hcA^\wk_{v;d}\subset\hcA^\wk_v$ the vector subspace of elements of degree $d$. 

\medskip

Various definitions, which we gave for the algebra of differential polynomials, can be immediately extended to the algebra $\hcA^\wkt_v$.
\begin{itemize}
\item A variational derivative $\frac{\delta}{\delta v^\alpha}\colon\hcA^\wkt_v\to\hcA^\wkt_v$ is defined in the same way. There is an inclusion $\Im(\d_x)\subset\Ker(\frac{\delta}{\delta v^\alpha})$.

\smallskip

\item There are spaces $\Lambda^\wkt_v:=\cA^\wkt_v/(\Im(\d_x)\oplus\mbC)$ and $\hLambda^\wkt_v:=\Lambda^\wkt_v[[\eps]]$. Variational derivatives $\frac{\delta}{\delta v^\alpha}\colon \hLambda^\wkt_v\to\hcA^\wkt_v$, $\alpha\in[N]$, are well defined. Using the algebra $\hcA^\wkt_v$, matrix differential operators are defined in the same way, and also the associated brackets on the space~$\hLambda^\wkt_v$. The associated space of Poisson operators will be denoted by $\mathcal{PO}^\wkt_v$.

\smallskip

\item Using Poisson brackets on $\hLambda^\wkt_v$, we define Hamiltonian systems of PDEs.

\smallskip

\item We will consider changes of variables of the form
\begin{gather}\label{eq:weak Miura transformation}
v^\alpha\mapsto u^\alpha(v^*_*,\eps)=f^\alpha_0+\sum_{k\ge 1}\eps^k f^\alpha_k,\quad 
\begin{minipage}{9cm}
$f^\alpha_0\in\mbC[[v^1,\ldots,v^N]],\quad f^\alpha_k\in\cA^\wkt_v$ for $k\ge 1$,\\
$f^\alpha_0|_{v^*=0}=0,\quad\left.\det\left(\frac{\d f^\alpha_0}{\d v^\beta}\right)\right|_{v^*=0}\ne 0$,
\end{minipage}
\end{gather}
and call them \emph{weak Miura transformations}. These transformations form a group.
\end{itemize}

\medskip

\begin{remark}\label{remark:inverse to weak Miura}
It is useful to note that for a weak Miura transformation~\eqref{eq:weak Miura transformation} the terms $h^\alpha_1,h^\alpha_2,\ldots$ of the inverse weak Miura transformation
$$
u^\alpha\mapsto v^\alpha(u^*_*,\eps)=h^\alpha_0+\sum_{k\ge 1}\eps^k h^\alpha_k(u^*_*)
$$
can be recursively determined by the relation
\begin{gather}\label{eq:resursion for inverse weak Miura}
\left.h^\alpha_{k+1}(u^*_*)\right|_{u^\gamma_c\mapsto\d_x^c f^\gamma_0(v^*_*)}+\Coef_{\eps^{k+1}}\lb\sum_{i=0}^k \left.h^\alpha_i(u^*_*)\eps^i\right|_{u^\gamma_c\mapsto\d_x^c u^\gamma(v^*_*,\eps)}\rb=0,\quad k\ge 0.
\end{gather}
\end{remark}

\medskip

\begin{remark}
We see that various notions that we consider in this paper, like local functional, Poisson operator, Hamiltonian system, can be based on the algebra of differential polynomials~$\hcA_v$ or on the bigger algebra~$\hcA^\wkt_v$. Let us agree that when we add the adjective ``polynomial'', then this signals that the corresponding notion is based on the algebra of differential polynomials. 
\end{remark}

\medskip

Note that for any $i\in\mbZ_{\ge 1}$ the fraction $\frac{1}{(v^1_x)^i}$ can be naturally considered as an element of the algebra~$\cA^\wkt_v$:
$$
\frac{1}{(v^1_x)^i}=(1+(v^1_x-1))^{-i}=\sum_{k\ge 0}{-i\choose k}(v^1_x-1)^k,
$$
and clearly $\deg\lb\frac{1}{(v^1_x)^i}\rb=-i$. For any $d\in\mbZ$ let~$\cA^\rt_{v;d}$ be the vector subspace of $\cA^\wk_v$ consisting of elements of the form
\begin{gather*}
\sum_{i\le m}P_i(v^1_x)^i,\quad m\in\mbZ,\quad P_i\in\cA_{v;d-i},\quad \frac{\d P_i}{\d v^1_x}=0.
\end{gather*}
Clearly $\cA^\rt_{v;d}\subset\cA^\wk_{v;d}$. Let 
$$
\cA^{\rt}_v:=\bigoplus_{d\in\mbZ}\cA^\rt_{v;d}\subset\cA^\wk_v,\qquad \cA^\rtt_v:=\cA^\rt_v\cap\cA^\wkt_v\subset\cA^\rt_v.
$$
Define extended spaces 
$$
\hcA^\rt_v:=\cA^\rt_v[[\eps]],\qquad\hcA^\rtt_v:=\cA^\rtt_v[[\eps]].
$$

\medskip

We will consider changes of variables of the form
$$
v^\alpha\mapsto u^\alpha(v^*_*,\eps)=v^\alpha+\sum_{k\ge 1}\eps^k f^\alpha_k,\quad 
f^\alpha_k\in\cA^\rtt_{v;k}.
$$
and call them \emph{rational Miura transformations}. They form a subgroup in the group of weak Miura transformations. For an element $f=\sum_{i\le m}P_i(v^1_x)^i\in\cA^\rt_v$, $\frac{\d P_i}{\d v^1_x}=0$, define the \emph{polynomial} and the \emph{singular} parts by
\begin{gather*}
f^\pol:=\sum_{i=0}^m P_i(v^1_x)^i\in\cA_v,\qquad f^\sing:=f-f^\pol\in\cA^\rt_v.
\end{gather*}
A \emph{purely singular Miura transformation} is a rational Miura transformation $v^\alpha\mapsto u^\alpha(v^*_*,\eps)$ satisfying the conditions
\begin{gather}\label{eq:conditions for purely singular}
u^\alpha(v^*_*,\eps)^\pol=v^\alpha,\qquad \frac{\d u^\alpha(v^*_*,\eps)}{\d v^1}=\delta^{\alpha,1}.
\end{gather}

\medskip

\subsection{Technical results}

\begin{lemma}\label{lemma:singular-1}{\ }
\begin{enumerate}
\item For any purely singular Miura transformation $v^\alpha\mapsto u^\alpha(v^*_*,\eps)$ and $f\in\hcA^\rtt_u$, we have
\begin{gather*}
\lb f|_{u^\gamma_c\mapsto \d_x^c u^\gamma(v^*_*,\eps)}\rb^\pol=\lb(f^\pol)_{u^\gamma_c\mapsto \d_x^c u^\gamma(v^*_*,\eps)}\rb^\pol.
\end{gather*}
If, moreover, the condition $\frac{\d f^\pol}{\d u^1_x}=0$ is satisfied, then we have
\begin{gather*}
\lb f|_{u^\gamma_c\mapsto \d_x^c u^\gamma(v^*_*,\eps)}\rb^\pol=(f^\pol)_{u^\gamma_c\mapsto v^\gamma_c}.
\end{gather*}

\smallskip

\item Purely singular Miura transformations form a subgroup in the group of all rational Miura transformations.

\smallskip

\end{enumerate}
\end{lemma}
\begin{proof}
{\it 1}. Note that $\lb\d_x^c u^\alpha(v^*_*,\eps)\rb^\pol=v^\alpha_c$, which implies that if a differential polynomial $h\in\cA_u$ satisfies $\frac{\d h}{\d u^1_x}=0$, then $\lb h|_{u^\gamma_c\mapsto \d_x^c u^\gamma(v^*_*,\eps)}\rb^\pol=h|_{u^\gamma_c\mapsto v^\gamma_c}$ and $\lb\left.\frac{h}{(u^1_x)^k}\right|_{u^\gamma_c\mapsto \d_x^c u^\gamma(v^*_*,\eps)}\rb^\pol=0$, where $k\ge 1$. From this Part 1 immediately follows.

\medskip

{\it 2}. The fact that the composition of purely singular Miura transformations satisfies the first condition in~\eqref{eq:conditions for purely singular} follows from Part~1. The second condition follows from the chain rule. Using Remark~\ref{remark:inverse to weak Miura}, we see that the inverse to a purely singular Miura transformation is recursively determined by the relation
\begin{gather*}
\left.h^\alpha_{k+1}(u^*_*)\right|_{u^\gamma_c\mapsto v^\gamma_c}+\Coef_{\eps^{k+1}}\lb u^\alpha(v^*_*)+\sum_{i=1}^k \left.h^\alpha_i(u^*_*)\eps^i\right|_{u^\gamma_c\mapsto\d_x^c u^\gamma(v^*_*,\eps)}\rb=0,\quad k\ge 0.
\end{gather*}
Then both conditions in~\eqref{eq:conditions for purely singular} are easily proved by the induction on~$k$, using Part~1 and the chain rule.
\end{proof}

\medskip

For a differential operator $K=\sum_{i\ge 0}f_i\d_x^i$, $f_i\in\hcA^\rtt_v$, we denote $K^\pol:=\sum_{i\ge 0}f_i^\pol\d_x^i$. 

\medskip

\begin{lemma}\label{lemma:singular-2}
Consider a purely singular Miura transformation $v^\alpha\mapsto u^\alpha(v^*_*,\eps)$ and a matrix differential operator $K=(K^{\alpha\beta})$ of the form
$$
K^{\alpha\beta}=g^{\alpha\beta}\d_x+b^{\alpha\beta}_\gamma v^\gamma_x,\quad g^{\alpha\beta},b^{\alpha\beta}_\gamma\in\mbC[[v^1,\ldots,v^N]],\quad \frac{\d^2 g^{\alpha\beta}}{(\d v^1)^2}=\frac{\d b^{\alpha\beta}_\gamma}{\d v^1}=0,\quad b^{\alpha\beta}_1\in\mbC.
$$
Then 
$$
(K_u^{\alpha\beta})^\pol=\left.\lb K^{\alpha\beta}+L_\mu^1(P^\alpha)\circ\eta^{\mu\beta}\d_x+L_\mu(P^\alpha)\circ b^{\mu\beta}_1+b^{\alpha\nu}_1\circ L_\nu(P^\beta)^\dagger\rb\right|_{v^\gamma_n\mapsto u^\gamma_n},
$$
where $\eta^{\alpha\beta}:=\frac{\d g^{\alpha\beta}}{\d v^1}$ and $P^\alpha:=\lb v^1_x (u^\alpha(v^*_*,\eps)-v^\alpha)\rb^\pol$.
\end{lemma}
\begin{proof}
Define $\tu^\alpha(v^*_*,\eps):=u^\alpha(v^*_*,\eps)-v^\alpha$. We have $K_u^{\alpha\beta}=\left.\tK^{\alpha\beta}\right|_{v^\gamma_c\mapsto\d_x^c v^\gamma(u^*_*,\eps)}$, where
$$
\tK^{\alpha\beta}:=\lb\delta^\alpha_\mu+\sum_{p\ge 0}\frac{\d\tu^\alpha(v^*_*,\eps)}{\d v^\mu_p}\d_x^p\rb\circ\lb g^{\mu\nu}\d_x+b^{\mu\nu}_\xi v^\xi_x\rb\circ\lb\delta^\beta_\nu+\sum_{q\ge 0}(-\d_x)^q\circ\frac{\d\tu^\beta(v^*_*,\eps)}{\d v^\nu_q}\rb.
$$
By Lemma~\ref{lemma:singular-1}, we have $\lb\left.\tK^{\alpha\beta}\right|_{v^\gamma_c\mapsto\d_x^c v^\gamma(u^*_*,\eps)}\rb^\pol=\lb\left.(\tK^{\alpha\beta})^\pol\right|_{v^\gamma_c\mapsto\d_x^c v^\gamma(u^*_*,\eps)}\rb^\pol$.
Let us now compute~$(\tK^{\alpha\beta})^\pol$.

\medskip

We have
\begin{align*}
\tK^{\alpha\beta}=&g^{\alpha\beta}\d_x+b^{\alpha\beta}_\xi v^\xi_x+\underbrace{\lb g^{\alpha\nu}\d_x+b^{\alpha\nu}_\xi v^\xi_x\rb\circ\sum_{q\ge 0}(-\d_x)^q\circ\frac{\d\tu^\beta(v^*_*,\eps)}{\d v^\nu_q}}_{A:=}\\
&+\underbrace{\sum_{p\ge 0}\frac{\d\tu^\alpha(v^*_*,\eps)}{\d v^\mu_p}\d_x^p\circ\lb g^{\mu\beta}\d_x+b^{\mu\beta}_\xi v^\xi_x\rb}_{B:=}\\
&+\underbrace{\lb\sum_{p\ge 0}\frac{\d\tu^\alpha(v^*_*,\eps)}{\d v^\mu_p}\d_x^p\rb\circ\lb g^{\mu\nu}\d_x+b^{\mu\nu}_\xi v^\xi_x\rb\circ\lb\sum_{q\ge 0}(-\d_x)^q\circ\frac{\d\tu^\beta(v^*_*,\eps)}{\d v^\nu_q}\rb}_{C:=}.
\end{align*}
Denote 
$$
\hcA^{\rtt,\sing}_v:=\left\{f\in\hcA^\rtt_v\left|f^\pol=0,\frac{\d f}{\d v^1}=0\right.\right\}\subset\hcA^\rtt_v.
$$
This subspace is closed under the multiplication and also satisfies $\d_x \hcA^{\rtt,\sing}_v\subset\hcA^{\rtt,\sing}_v$. Note that
$$
\tu^\alpha(v^*_*,\eps)-\frac{P^\alpha}{v^1_x}\in\frac{1}{v^1_x}\hcA^{\rtt,\sing}_v,\quad\frac{\d P^\alpha}{\d v^1}=\frac{\d P^\alpha}{\d v^1_x}=0,
$$
which implies that
\begin{gather}\label{eq:initial term of tu}
\d_x^k\frac{\d\tu^\alpha(v^*_*,\eps)}{\d v^\mu_p}-\frac{\d_x^k\frac{\d P^\alpha}{\d v^\mu_p}}{v^1_x}\in\frac{1}{v^1_x}\hcA^{\rtt,\sing}_v,\quad k\ge 0.
\end{gather}
Let us now compute $A^\pol$, $B^\pol$, and $C^\pol$.

\medskip

Using~\eqref{eq:initial term of tu}, we see that
\begin{align*}
&\left(g^{\alpha\nu}\d_x\circ\sum_{q\ge 0}(-\d_x)^q\circ\frac{\d\tu^\beta(v^*_*,\eps)}{\d v^\nu_q}\right)^\pol=0,\\
&\left(b^{\alpha\nu}_\xi v^\xi_x\circ\sum_{q\ge 0}(-\d_x)^q\circ\frac{\d\tu^\beta(v^*_*,\eps)}{\d v^\nu_q}\right)^\pol=\left(b^{\alpha\nu}_\xi v^\xi_x\circ\frac{1}{v^1_x}\sum_{q\ge 0}(-\d_x)^q\circ\frac{\d P^\beta}{\d v^\nu_q}\right)^\pol=b^{\alpha\nu}_1\circ L_\nu(P^\beta)^\dagger.
\end{align*}
Therefore, $A^\pol=b^{\alpha\nu}_1\circ L_\nu(P^\beta)^\dagger$. Using the properties~$\frac{\d^2 g^{\alpha\beta}}{(\d v^1)^2}=\frac{\d b^{\alpha\beta}_\gamma}{\d v^1}=0$ and~\eqref{eq:initial term of tu}, we obtain
$$
B^\pol=\left(\frac{1}{v^1_x}L_\mu(P^\alpha)\circ(\eta^{\mu\beta}v^1\d_x+b^{\mu\beta}_1 v^1_x)\right)^\pol=L_\mu^1(P^\alpha)\circ\eta^{\mu\beta}\d_x+L_\mu(P^\alpha)\circ b^{\mu\beta}_1,
$$
and $C^\pol=0$.

\medskip

Thus, we have
\begin{align*}
&\lb K^{\alpha\beta}_u\rb^\pol=b^{\alpha\beta}_1 \lb\d_x v^1(u^*_*,\eps)\rb^\pol\\
&+\lb\left.\Big(\underline{g^{\alpha\beta}\d_x+\sum_{\xi\ne 1}b^{\alpha\beta}_\xi v^\xi_x+L_\mu^1(P^\alpha)\circ\eta^{\mu\beta}\d_x+L_\mu(P^\alpha)\circ b^{\mu\beta}_1+b^{\alpha\nu}_1\circ L_\nu(P^\beta)^\dagger}\Big)\right|_{v^\gamma_c\mapsto \d_x^c v^\gamma(u^*_*,\eps)}\rb^\pol,
\end{align*}
where used that $b^{\alpha\beta}_1\in\mbC$. The first term on the right-hand side of this equality is equal to
$b^{\alpha\beta}_1 u^1_x(u^*_*,\eps)$. In the second term, note that the derivative by $v^1_x$ of each coefficient of the underlined operator is equal to~$0$, so by Lemma~\ref{lemma:singular-1} the second term is equal to 
$$
\left.\lb g^{\alpha\beta}\d_x+\sum_{\xi\ne 1}b^{\alpha\beta}_\xi v^\xi_x+L_\mu^1(P^\alpha)\circ\eta^{\mu\beta}\d_x+L_\mu(P^\alpha)\circ b^{\mu\beta}_1+b^{\alpha\nu}_1\circ L_\nu(P^\beta)^\dagger\rb\right|_{v^\gamma_n\mapsto u^\gamma_n},
$$
which completes the proof of the lemma.
\end{proof}

\medskip

\subsection{Relation with the string and the dilaton equations}

Consider formal variables $t^\alpha_a$, $\alpha\in[N]$, $a\in\mbZ_{\ge 0}$. For $d\in\mbZ_{\ge 0}$, denote by $\mbC[[t^*_*]]^{(d)}\subset\mbC[[t^*_*]]$ the ideal consisting of infinite linear combinations of monomials $\prod t^{\alpha_i}_{d_i}$ satisfying $\sum d_i\ge d$. 

\medskip

Consider formal power series $\tv^\alpha\in\mbC[[t^*_*]]$, $\alpha\in[N]$, satisfying 
\begin{gather}\label{eq:conditions for tv}
\left.\tv^\alpha\right|_{t^*_{\ge 1}=0}=t^\alpha_0,\qquad\underbrace{\left(\frac{\d}{\d t^1_0}-\sum_{k\ge 0}t^\alpha_{k+1}\frac{\d}{\d t^\alpha_k}\right)}_{\cS:=}\tv^\alpha=\delta^{\alpha,1},\qquad
\underbrace{\left(\frac{\d}{\d t^1_1}-\sum_{k\ge 0}t^\alpha_k\frac{\d}{\d t^\alpha_k}\right)}_{\cD:=}\tv^\alpha=0.
\end{gather}
The last two equations are usually called the \emph{string} and the \emph{dilaton} equation, respectively. Denote $\tv^\alpha_n:=\frac{\d^n\tv^\alpha}{(\d t^1_0)^n}$. From the string equation it follows that
$$
\tv^\alpha_n=\delta_{n,1}\delta^{\alpha,1}+\left(\sum_{k\ge 0}t^\gamma_{k+1}\frac{\d}{\d t^\gamma_k}\right)^n\tv^\alpha,
$$
and therefore
\begin{gather}\label{eq:property of tv}
\tv^\alpha_n-t^\alpha_n-\delta_{n,1}\delta^{\alpha,1}\in\mbC[[t^*_*]]^{(n+1)}.
\end{gather}
This implies that for any element $P\in\cA^\wk_v$ the substitution $P|_{v^\gamma_n\mapsto \tv^\gamma_n}\in\mbC[[t^*_*]]$ is well defined and, moreover, it defines an algebra isomorphism $\cA^\wk_v\to\mbC[[t^*_*]]$. We will say that this isomorphism is given by the $N$-tuple of formal power series $(\tv^1,\ldots,\tv^N)$. The preimage of an element $f\in\mbC[[t^*_*]]$ under this isomorphism will be denoted by~$f(v^*_*)$. 

\medskip

Obviously, $\frac{\d}{\d t^1_0}\tv^\alpha_n=\tv^\alpha_{n+1}$. Also, from the string and the dilaton equations it follows that $\cS(\tv^\alpha_n)=\delta^{\alpha,1}\delta_{n,0}$ and $\cD(\tv^\alpha_n)=n \tv^\alpha_n$. This implies that under the isomorphism $\cA^\wk_v\cong\mbC[[t^*_*]]$ the operators $\frac{\d}{\d t^1_0}$, $\cS$, and $\cD$ on $\mbC[[t^*_*]]$ correspond to the operator $\d_x$, $\frac{\d}{\d v^1_0}$, and $\cD_v:=\sum_{k\ge 0}k v^\gamma_k\frac{\d}{\d v^\gamma_k}$ on $\cA^\wk_v$, respectively.

\medskip

\begin{lemma}\label{lemma:leading term of rational function}
Consider a formal power series $f\in\mbC[[t^*_*]]$ satisfying $\cD(f)=d\cdot f$, $d\in\mbZ$.
\begin{enumerate}
\item We have $f(v^*_*)\in\cA^\rt_{v;d}$.

\smallskip

\item Suppose $f\in\mbC[[t^*_*]]^{(k)}$ for some $k\in\mbZ_{\ge 0}$. Let us decompose $f=\sum_{i\ge k}h_i$, where $h_i\in\mbC[[t^*_*]]$ is homogeneous of degree $i$, if we assign to $t^\gamma_c$ degree $c$. Then $f(v^*_*)$ has the form
$$
f(v^*_*)=\sum_{i\ge k}P_i(v^1_x)^{d-i},\quad P_i\in\cA_{v;i},\quad\frac{\d P_i}{\d v^1_x}=0,
$$
where $P_{k}=\left.h_k\right|_{t^\gamma_n\mapsto v^\gamma_n}$.
\end{enumerate}
\end{lemma}
\begin{proof}
We proceed in the same way as in the proof of Proposition~7.6 from the paper~\cite{BDGR20}. Suppose $f\in\mbC[[t^*_*]]^{(k)}$ for some $k\in\mbZ_{\ge 0}$. We construct a sequence of pairs $(f_i,Q_i)\in\mbC[[t^*_*]]^{(k+i)}\times\cA_{v;k+i-1}$, $i\ge 0$, with $f_0=f$ and $Q_0=0$, by the following recursive formulas:
\begin{align*}
&Q_i:=\sum_{n\ge 0}\sum_{\substack{\alpha_1,\ldots,\alpha_n\in[N]\\d_1\,\ldots,d_n\in\mbZ_{\ge 0}\\\sum d_j=k+i-1}}\left.\frac{\d^n f_{i-1}}{\d t^{\alpha_1}_{d_1}\cdots\d t^{\alpha_n}_{d_n}}\right|_{t^*_*=0}\frac{\prod v^{\alpha_j}_{d_j}}{n!},\quad f_i:=f_{i-1}-\left.\lb(v^1_x)^{d+1-k-i}Q_i\rb\right|_{v^\gamma_c\mapsto\tv^\gamma_c},\quad i\ge 1.
\end{align*}
Let us check that $\cD(f_i)=d\cdot f_i$ and $f_i\in\mbC[[t^*_*]]^{(k+i)}$ by the induction on $i$. Indeed, if $\cD(f_{i-1})=d\cdot f_{i-1}$ and $f_{i-1}\in\mbC[[t^*_*]]^{(k+i-1)}$, then we immediately obtain that $\left.\frac{\d^n f_{i-1}}{\d t^{\alpha_1}_{d_1}\cdots\d t^{\alpha_n}_{d_n}}\right|_{t^*_*=0}=0$ if $\sum d_j=k+i-1$ and $(\alpha_l,d_l)=(1,1)$ for some $l\in[n]$. Therefore, $\frac{\d Q_i}{\d v^1_x}=0$, and the property~\eqref{eq:property of tv} implies that $f_i\in\mbC[[t^*_*]]^{(k+i)}$. Using the formula for~$f_i$, we also see that $\cD(f_i)=d\cdot f_i$. From the above construction, we obtain
$$
f=\left.\sum_{i\ge 1}Q_i(v^1_x)^{d+1-k-i}\right|_{v^\gamma_c\mapsto\tv^\gamma_c},\quad Q_i\in\cA_{v;i+k-1},\quad\frac{\d Q_i}{\d v^1_x}=0,
$$
which implies the first part of the lemma. The second part follows from the construction of~$Q_1$.
\end{proof}

\medskip


\section{Dubrovin--Zhang hierarchy}\label{section:DZ hierarchy}

Consider a semisimple CohFT. Its \emph{correlators} are defined by
$$
\<\tau_{d_1}(v_1)\cdots\tau_{d_n}(v_n)\>_g:=\int_{\oM_{g,n}}c_{g,n}(\otimes_{i=1}^n v_i)\prod_{i=1}^n\psi_i^{d_i},\quad v_1,\ldots,v_n\in V,\quad d_1,\ldots,d_n\in\mbZ_{\ge 0}.
$$
Choosing a basis $e_1,\ldots,e_N\in V$ such that $e_1=e$, consider the \emph{potential} of our CohFT:
\begin{gather*}
\cF(t^*_*,\eps)=\sum_{g\ge 0}\eps^{2g}\cF_g(t^*_*):=\sum_{g,n\ge 0}\frac{\eps^{2g}}{n!}\sum_{\substack{\alpha_1,\ldots,\alpha_n\in[N]\\d_1,\ldots,d_n\in\mbZ_{\ge 0}}}\<\prod_{i=1}^n\tau_{d_i}(e_{\alpha_i})\>_g\prod_{i=1}^n t^{\alpha_i}_{d_i}\in\mbC[[t^*_*,\eps]].
\end{gather*}
The potential $\mcF$ satisfies the \emph{string} and the \emph{dilaton} equations:
\begin{gather*}
\cS(\cF)=\frac{1}{2}\eta_{\alpha\beta}t^\alpha_0 t^\beta_0+\eps^2\<\tau_0(e_1)\>_1,\qquad \cD(\cF)=\eps\frac{\d\mcF}{\d\eps}-2\mcF+\eps^2\frac{N}{24}.
\end{gather*}

\medskip

Define local functionals
$$
\oh_{\alpha,d}^{\DZ;[0]}:=\int\left.\frac{\d\mcF_0}{\d t^\alpha_d}\right|_{t^\gamma_n\mapsto\delta_{n,0}v^\gamma}dx\in\Lambda_{v;0},\quad\alpha\in[N],\quad d\in\mbZ_{\ge 0}.
$$
We consider the matrix differential operator $\eta^{-1}\d_x$ as a Poisson operator from $\mathcal{PO}_v$. The \emph{principal hierarchy} is the Hamiltonian hierarchy defined by the Poisson operator $\eta^{-1}\d_x$ and the local functionals $\oh^{\DZ;[0]}_{\alpha,d}$. The principal hierarchy is integrable (see, e.g.,~\cite[Proposition~7]{BPS12}). We additionally equip the principal hierarchy with the local functionals $\oh_{\alpha,-1}^{\DZ;[0]}:=\int\eta_{\alpha\nu}v^\nu dx$, $\alpha\in[N]$, which are Casimirs of the Poisson bracket $\{\cdot,\cdot\}_{\eta^{-1}\d_x}$. From the string equation for~$\mcF$, it follows that $\oh^{\DZ;[0]}_{1,0}=\int\frac{1}{2}\eta_{\alpha\beta}v^\alpha v^\beta dx$, so the flow $\frac{\d}{\d t^1_0}$ of the principal hierarchy is given by $\frac{\d v^\alpha}{\d t^1_0}=v^\alpha_x$.

\medskip

For $\alpha\in[N]$ and $n\in\mbZ_{\ge 0}$, introduce formal power series
$$
w^{\top;\alpha}:=\eta^{\alpha\mu}\frac{\d^2\cF}{\d t^\mu_0\d t^1_0},\quad w^{\top;\alpha}_n:=\frac{\d^n w^{\top;\alpha}}{(\d t^1_0)^n},\quad v^{\top;\alpha}:=w^{\top;\alpha}|_{\eps=0},\quad v^{\top;\alpha}_n:=w^{\top;\alpha}_n|_{\eps=0}.
$$
From the string and the dilaton equations for~$\cF$, it follows that $\tv^\alpha=v^{\top;\alpha}$ satisfies the conditions in~\eqref{eq:conditions for tv}. So we consider the isomorphism $\hcA^\wk_v\cong\mbC[[t^*_*,\eps]]$ given by the $N$-tuple $\ov^\top:=(v^{\top;1},\ldots,v^{\top;N})$. Note also that $\ov^\top$ is a solution of the principal hierarchy~(see, e.g.,~\cite[Section~3.1.2]{BPS12}).

\medskip

Let us consider an element $w^\alpha(v^*_*,\eps):=w^\top(v^*_*,\eps)\in\hcA^\wk_v$. Clearly, $w^\alpha(v^*_*,\eps)=v^\alpha+O(\eps)$. In~\cite[Proposition~7.6]{BDGR20}, the authors proved that $w^\alpha(v^*_*,\eps)\in\hcA^\rtt_{v;0}$, so the change of variables $v^\alpha\mapsto w^\alpha(v^*_*,\eps)$ is a rational Miura transformation.

\medskip

Consider the local functionals $\oh_{\alpha,d}^{\DZ}:=\oh_{\alpha,d}^{\DZ;[0]}[w]\in\hLambda^\rtt_{w;0}$, $\alpha\in[N]$, $d\in\mbZ_{\ge -1}$, and the Poisson operator $K^\DZ_{(1)}:=(\eta^{-1}\d_x)_w\in\mathcal{PO}^\rtt_w$. In~\cite[Section~7]{BPS12}, the authors proved these local functionals and the Poisson operator are actually polynomial:
$$
\oh_{\alpha,d}^{\DZ}\in\hLambda_{w;0},\qquad K^\DZ_{(1)}\in\mathcal{PO}_w.
$$
Since $\mcF_g(v^*_*)\in\cA^\wkt_v$ \cite[Proposition~4]{BPS12} for $g\ge 1$, we have 
$$
\Coef_{\eps^{2g}}w^\alpha(v^*_*,\eps)=\eta^{\alpha\nu}\frac{\d^2\cF_g}{\d t^1_0\d t^\nu_0}(v^*_*)=\d_x\left\{\cF_g(v^*_*),\eta^{\alpha\nu}\oh_{\nu,0}^{\DZ;[0]}\right\}_{\eta^{-1}\d_x}\in\Im\{\d_x\colon\cA^\wkt_v\to\cA^\wkt_v\},\, g\ge 1.
$$
Therefore, $\oh^{\DZ}_{\alpha,-1}=\int\eta_{\alpha\nu}w^\nu dx$.

\medskip

\begin{definition}
The Hamiltonian hierarchy defined by the Poisson operator $K^\DZ_{(1)}$ and the local functionals $\oh^\DZ_{\alpha,d}$, $\alpha\in[N]$, $d\in\mbZ_{\ge 0}$, is called the \emph{Dubrovin--Zhang hierarchy}.
\end{definition}

\medskip

Since $\left\{\oh^{\DZ;[0]}_{\alpha_1,d_1},\oh^{\DZ;[0]}_{\alpha_2,d_2}\right\}_{\eta^{-1}\d_x}=0$, we have $\left\{\oh^{\DZ}_{\alpha_1,d_1},\oh^{\DZ}_{\alpha_2,d_2}\right\}_{K^\DZ_{(1)}}=0$ and so the Dubrovin--Zhang hierarchy is integrable. The flow $\frac{\d}{\d t^1_0}$ of the Dubrovin--Zhang hierarchy is given by $\frac{\d w^\alpha}{\d t^1_0}=w^\alpha_x$. The $N$-tuple $\lb w^{\top;1},\ldots,w^{\top;N}\rb$ is a solution of the Dubrovin--Zhang hierarchy. It is called the \emph{topological solution}.

\medskip

Suppose now that our CohFT is homogeneous. Let us choose a homogeneous basis $e_1,\ldots,e_N\in V$ with $e_1=e$, $\deg(e_\alpha)=q_\alpha$. Define
$$
b^{\alpha\beta}:=\left.\eta^{\alpha\mu}\eta^{\beta\nu}\frac{\d^2 \cF_0}{\d t^\mu_0\d t^\nu_0}\right|_{t^\theta_k\mapsto \delta_{k,0}v^\theta}\in\mbC[[v^1,\ldots,v^N]].
$$
Define a matrix differential operator $K^{\DZ;[0]}_{(2)}=\lb K^{\DZ;[0],\alpha\beta}_{(2)}\rb$ by
\begin{gather*}
K^{\DZ;[0],\alpha\beta}_{(2)}:=g^{\alpha\beta}\d_x+\d_x b^{\alpha\beta}\left(\frac{1}{2}-\mu_\beta\right),\quad \text{where}\quad g^{\alpha\beta}:=((1-q_\nu)v^\nu+r^\nu)\frac{\d b^{\alpha\beta}}{\d v^\nu}.
\end{gather*}
This matrix differential operator is Poisson, it is compatible with the Poisson operator $\eta^{-1}\d_x$, and it endows the principal hierarchy with a bihamiltonian structure \cite[Theorem~6.3]{Dub96} (see also \cite[Proposition~2.1 and Remark~2.2]{BRS21}):
\begin{gather*}
\left\{\cdot,\oh^{\DZ;[0]}_{\alpha,d}\right\}_{K_{(2)}^{\DZ;[0]}}=\left(d+\frac{3}{2}+\mu_\alpha\right)\left\{\cdot,\oh^{\DZ;[0]}_{\alpha,d+1}\right\}_{\eta^{-1}\d_x}+A^\beta_\alpha\left\{\cdot,\oh^{\DZ;[0]}_{\beta,d}\right\}_{\eta^{-1}\d_x},\quad d\ge -1.
\end{gather*}

\medskip

In~\cite[Theorem~1.2]{LWZ21}, the authors proved that the Poisson operator $K^\DZ_{(2)}:=K^{\DZ;[0]}_{(2);w}\in\mathcal{PO}^\rtt_w$ is actually polynomial:
$$
K^\DZ_{(2)}\in\mathcal{PO}_w.
$$
Summarizing, the Dubrovin--Zhang hierarchy associated to our semisimple homogeneous CohFT is a polynomial bihamiltonian integrable hierarchy given by the local functionals $\oh_{\alpha,d}^\DZ$ and the Poisson operators $K^\DZ_{(1)}, K^\DZ_{(2)}$ with the bihamiltonian resursion relation
\begin{gather*}
\left\{\cdot,\oh^\DZ_{\alpha,d}\right\}_{K^\DZ_{(2)}}=\left(d+\frac{3}{2}+\mu_\alpha\right)\left\{\cdot,\oh^\DZ_{\alpha,d+1}\right\}_{K^\DZ_{(1)}}+A^\beta_\alpha\left\{\cdot,\oh^\DZ_{\beta,d}\right\}_{K^\DZ_{(1)}},\quad d\ge -1.
\end{gather*} 

\medskip


\section{Proof of Theorem~\ref{theorem:main}}\label{section:proof of the main theorem}

We consider an arbitrary semisimple homogeneous CohFT. It is easy to check that the validity of the theorem doesn't depend on the choice of a homogeneous basis in $V$. Let us choose a homogeneous basis such that $e=e_1$. 

\medskip

Consider the Dubrovin--Zhang hierarchy associated to our CohFT. As Hamiltonian hierarchies, the Dubrovin--Zhang hierarchy and the DR hierarchy are related by a Miura transformation $u^\alpha\mapsto w^\alpha(u^*_*,\eps)$ satisfying $w^\alpha(u^*_*,\eps)=u^\alpha+O(\eps)$ (see \cite[Proposition 7.4]{BDGR18}, \cite[Section~3.4]{BGR19},~and \cite[Theorem~1.2]{BLS24}):
$$
\eta^{-1}\d_x=\lb K^{\DZ}_{(1)}\rb_u,\qquad \og_{\alpha,d}=\oh^\DZ[u],\qquad \alpha\in[N],\,d\in\mbZ_{\ge -1}.
$$
Consider now the operator $\lb K^{\DZ}_{(2)}\rb_u\in\mathcal{PO}_u$. From the discussion in the previous section, we know that this operator is Poisson, it is compatible with the operator $\eta^{-1}\d_x$, and relation~\eqref{eq:bihamiltonian recursion for DR hierarchy} is satisfied if we replace $K^\DR$ by $\lb K^{\DZ}_{(2)}\rb_u\in\mathcal{PO}_u$. Therefore, it is sufficient to prove that
$$
K^\DR=\lb K^{\DZ}_{(2)}\rb_u.
$$ 

\medskip

Since the Dubrovin--Zhang hierarchy is obtained from the principal hierarchy by the rational Miura transformation $v^\alpha\mapsto w^\alpha(v^*_*,\eps)$, we see that we have a rational Miura transformation $v^\alpha\mapsto u^\alpha(v^*_*,\eps)$ relating the principal hierarchy and the DR hierarchy. So it is sufficient to prove the following proposition.

\medskip

\begin{proposition}
We have $K^{\DR}=\lb K^{\DZ;[0]}_{(2)}\rb_u$.
\end{proposition}
\begin{proof}

Denote $g^{\alpha,0}:=\eta^{\alpha\nu}g_{\nu,0}$.

\medskip

\begin{lemma}{\ }\label{lemma:singular-3}
\begin{enumerate}
\item The rational Miura transformation $v^\alpha\mapsto u^\alpha(v^*_*,\eps)$ is purely singular.

\smallskip

\item We have $\left(v^1_x\left(u^\alpha(v^*_*,\eps)-v^\alpha\right)\right)^\pol=\left.\d_x\lb g^{\alpha,0}-\left.g^{\alpha,0}\right|_{\eps=0}\rb\right|_{u^\gamma_c\mapsto v^\gamma_c}$.
\end{enumerate}
\end{lemma}
\begin{proof}
Since $\og_{1,0}=\int\frac{1}{2}\eta_{\alpha\beta}u^\alpha u^\beta$ \cite[Lemma~4.3]{Bur15}, the flow $\frac{\d}{\d t^1_0}$ of the DR hierarchy is given by $\frac{\d u^\alpha}{\d t^1_0}=u^\alpha_x$. Consider the solution $(u^{\str;1},\ldots,u^{\str;N})\in\mbC[[t^*_*,\eps]]^N$ of the DR hierarchy satisfying the initial condition $u^{\str;\alpha}|_{t^{\ne 1}_0=t^*_{\ge 1}=0}=\delta^{\alpha,1}t^1_0$ (see also Remark~\ref{remark:indentifying x and t}). This solution is called the \emph{string solution}. We denote $u^{\str;\alpha}_n:=\frac{\d^n u^{\str;\alpha}}{(\d t^1_0)^n}$.

\medskip

The Miura transformation relating the Dubrovin--Zhang hierarchy and the DR hierarchy maps the topological solution of the Dubrovin--Zhang hierarchy to the string solution of the DR hierarchy \cite[proof of Proposition~7.4]{BDGR18}:
$$
\left.w^\alpha(u^*_*,\eps)\right|_{u^\gamma_c\mapsto u^{\str;\gamma}_c}=w^{\top;\alpha}.
$$
Therefore, using the isomorphism $\hcA^\wk_v\cong\mbC[[t^*_*,\eps]]$ given by the $N$-tuple $\ov^\top$, the rational Miura transformation $v^\alpha\mapsto u^\alpha(v^*_*,\eps)$ can be described as
$$
u^\alpha(v^*_*,\eps)=u^{\str;\alpha}(v^*_*,\eps).
$$
The string solution of the DR hierarchy satisfies the string equation~\cite[Lemma~4.7]{Bur15}
$$
\cS\lb u^{\str;\alpha}\rb=\delta^{\alpha,1},
$$
which implies that $\frac{\d u^\alpha(v^*_*,\eps)}{\d v^1}$, and so the second condition in~\eqref{eq:conditions for purely singular} is satisfied.

\medskip

Let us now check the first condition in~\eqref{eq:conditions for purely singular}. Following~\cite[Section~4.4.2]{BS22}, we introduce~$N$ formal power series $\mcF^{\DR;\alpha}\in\mbC[[t^*_*,\eps]]$, $\alpha\in[N]$, by the formula
$$
\frac{\d\mcF^{\DR;\alpha}}{\d t^\beta_b}:=\left.\eta^{\alpha\nu}\frac{\delta\og_{\beta,b}}{\delta u^\nu}\right|_{u^\gamma_n\mapsto u^{\str;\gamma}_n},
$$
with the constant terms defined to be equal to zero, $\left.\mcF^{\DR;\alpha}\right|_{t^*_*=0}:=0$. Clearly, we have $u^{\str;\alpha}=\frac{\d\cF^{\DR;\alpha}}{\d t^1_0}$. Consider the expansion $\mcF^{\DR;\alpha}=\sum_{g\ge 0}\eps^{2g}\mcF^{\DR;\alpha}_g$. By~\cite[Theorem~4.9]{BS22}, we have
\begin{gather}\label{eq:vanishing for FDR}
\left.\frac{\d^n\mcF_g^{\DR;\alpha}}{\d t^{\alpha_1}_{d_1}\cdots\d t^{\alpha_n}_{d_n}}\right|_{t^*_*=0}=0,\quad\text{if}\quad\sum d_i\le 2g-1.
\end{gather}

\medskip

\begin{lemma}
The formal power series $\mcF^{\DR;\alpha}$ satisfies the string equation
\begin{gather}\label{eq:string for FDRalpha}
\cS\lb\cF^{\DR;\alpha}\rb=t^\alpha_0.
\end{gather}
\end{lemma}
\begin{proof}
Clearly, both sides of~\eqref{eq:string for FDRalpha} are zero if we set $t^*_*=0$. Recall that $\cS(u^{\str;\alpha})=\delta^{\alpha,1}$ and $\frac{\d\og_{\alpha,a}}{\d u^1}=\og_{\alpha,a-1}$, $a\ge 0$. Then, for $d\ge 1$, we compute
\begin{multline*}
\frac{\d}{\d t^\beta_d}\lb\cS\lb\cF^{\DR;\alpha}\rb\rb=\cS\lb\frac{\d\cF^{\DR;\alpha}}{\d t^\beta_d}\rb-\frac{\d\cF^{\DR;\alpha}}{\d t^\beta_{d-1}}=\\
=\eta^{\alpha\nu}\left[\cS\lb\left.\frac{\delta\og_{\beta,d}}{\delta u^\nu}\right|_{u^\gamma_n\mapsto u^{\str;\gamma}_n}\rb-\left.\frac{\delta\og_{\beta,d-1}}{\delta u^\nu}\right|_{u^\gamma_n\mapsto u^{\str;\gamma}_n}\right]=\left.\eta^{\alpha\nu}\left[\frac{\delta\frac{\d\og_{\beta,d}}{\d u^1}}{\delta u^\nu}-\frac{\delta\og_{\beta,d-1}}{\delta u^\nu}\right]\right|_{u^\gamma_n\mapsto u^{\str;\gamma}_n}=0.
\end{multline*}
For $d=0$, we have
\begin{multline*}
\frac{\d}{\d t^\beta_0}\lb\cS\lb\cF^{\DR;\alpha}\rb\rb=\cS\lb\frac{\d\cF^{\DR;\alpha}}{\d t^\beta_0}\rb
=\eta^{\alpha\nu}\cS\lb\left.\frac{\delta\og_{\beta,0}}{\delta u^\nu}\right|_{u^\gamma_n\mapsto u^{\str;\gamma}_n}\rb=\left.\eta^{\alpha\nu}\frac{\delta\frac{\d\og_{\beta,0}}{\d u^1}}{\delta u^\nu}\right|_{u^\gamma_n\mapsto u^{\str;\gamma}_n}=\delta^\alpha_\beta,
\end{multline*}
which completes the proof.
\end{proof}

\medskip

Using the lemma and the vanishing~\eqref{eq:vanishing for FDR}, we obtain
$$
\Coef_{\eps^{2g}}\left.\frac{\d^n u^{\str;\alpha}}{\d t^{\alpha_1}_{d_1}\cdots\d t^{\alpha_n}_{d_n}}\right|_{t^*_*=0}=0,\quad\text{if $g\ge 1$ and $\sum d_i\le 2g$}.
$$
By Lemma~\ref{lemma:leading term of rational function}, this implies that the first condition in~\eqref{eq:conditions for purely singular} is satisfied.

\medskip

Let us now prove Part~2 of the lemma. By~\cite[Theorem~4.9]{BS22}, in the case $\sum d_i=2g$ we have
\begin{align*}
\left.\frac{\d^n\mcF_g^{\DR;\alpha}}{\d t^{\alpha_1}_{d_1}\cdots\d t^{\alpha_n}_{d_n}}\right|_{t^*_*=0}=&\eta^{\alpha\nu}\Coef_{a_1^{d_1}\cdots a_n^{d_n}}\int_{\oM_{g,n+1}}\lambda_g\DR_g(-\sum a_i,a_1,\ldots,a_n)c_{g,n+1}(e_\nu\otimes\otimes_{i=1}^n e_{\alpha_i})=\\
=&\Coef_{\eps^{2g}}\frac{\d^n g^{\alpha,0}}{\d u^{\alpha_1}_{d_1}\cdots\d u^{\alpha_n}_{d_n}}.
\end{align*}
Using the string equation for $\cF_g^{\DR;\alpha}$, we obtain
$$
\Coef_{\eps^{2g}}\left.\frac{\d^n u^{\str;\alpha}}{\d t^{\alpha_1}_{d_1}\cdots\d t^{\alpha_n}_{d_n}}\right|_{t^*_*=0}=\Coef_{\eps^{2g}}\frac{\d^n(\d_x g^{\alpha,0})}{\d u^{\alpha_1}_{d_1}\cdots\d u^{\alpha_n}_{d_n}},\quad\text{if $g\ge 1$ and $\sum d_i=2g+1$}.
$$
By Lemma~\ref{lemma:leading term of rational function}, this implies Part~2 of the lemma.
\end{proof}

\medskip

Denote $\og_{[0]}:=\og|_{\eps=0}$ and $g^{\alpha,0}_{[0]}:=g^{\alpha,0}|_{\eps=0}$. Note that $g_{\alpha,0}=\frac{\delta\og}{\delta u^\alpha}$. The operator $K=K^{\DZ;[0]}_{(2)}$ satisfies the conditions of Lemma~\ref{lemma:singular-2}, with $b^{\alpha\beta}_1=\eta^{\alpha\beta}\lb\frac{1}{2}-\mu_\beta\rb$. Denote $A^{\alpha\beta}:=\eta^{\alpha\gamma}\eta^{\beta\nu}A_{\gamma\nu}$. Using Lemma~\ref{lemma:singular-2} and Remark~\ref{remark:more on KDR}, we compute
\begin{align*}
&\lb K_{(2)}^{\DZ;[0]}\rb_u^{\alpha\beta}=\lb\lb K_{(2)}^{\DZ;[0]}\rb_u^{\alpha\beta}\rb^\pol=\\
=&\left.\lb K_{(2)}^{\DZ;[0],\alpha\beta}\rb\right|_{v^\gamma_n\mapsto u^\gamma_n}+L_\mu^1\lb\d_x(g^{\alpha,0}-g^{\alpha,0}_{[0]})\rb\circ\eta^{\mu\beta}\d_x\\
&+L_\mu\lb\d_x(g^{\alpha,0}-g^{\alpha,0}_{[0]})\rb\circ b^{\mu\beta}_1+b^{\alpha\nu}_1\circ L_\nu\lb\d_x(g^{\beta,0}-g^{\beta,0}_{[0]})\rb^\dagger=\\
=&\cancel{\d_x\circ\hOmega(\og_{[0]})^{\alpha\beta}\circ\left(\frac{1}{2}-\mu_\beta\right)}+\left(\frac{1}{2}-\mu_\alpha\right)\circ\hOmega(\og_{[0]})^{\alpha\beta}\circ\d_x+A^{\alpha\beta}\d_x+\\
&+L_\mu^1\lb\d_x (g^{\alpha,0}-g^{\alpha,0}_{[0]})\rb\circ\eta^{\mu\beta}\d_x+\d_x\circ L_\mu\lb g^{\alpha,0}-\cancel{g^{\alpha,0}_{[0]}}\rb\circ b^{\mu\beta}_1-b^{\alpha\nu}_1\circ L_\nu\lb g^{\beta,0}-g^{\beta,0}_{[0]}\rb^\dagger\circ\d_x=\\
=&\cancel{\left(\frac{1}{2}-\mu_\alpha\right)\circ\hOmega(\og_{[0]})^{\alpha\beta}\circ\d_x}+A^{\alpha\beta}\d_x+\d_x\circ L_\mu^1(g^{\alpha,0})\circ\eta^{\mu\beta}\d_x+L_\mu\lb g^{\alpha,0}-\cancel{g^{\alpha,0}_{[0]}}\rb\circ\eta^{\mu\beta}\d_x\\
&+\d_x\circ L_\nu(g^{\alpha,0})\circ \eta^{\nu\beta}\lb\frac{1}{2}-\mu_\beta\rb-\eta^{\alpha\nu}\lb\frac{1}{2}-\mu_\nu\rb\circ L_\nu\lb g^{\beta,0}-\cancel{g^{\beta,0}_{[0]}}\rb^\dagger\circ\d_x=\\
=&A^{\alpha\beta}\d_x+\d_x\circ\hOmega^1(\og)^{\alpha\beta}\circ\d_x+\hOmega(\og)^{\alpha\beta}\circ\d_x+\d_x\circ\hOmega(\og)^{\alpha\beta}\circ\lb\frac{1}{2}-\mu_\beta\rb-\lb\frac{1}{2}+\mu_\alpha\rb\circ\lb\hOmega(\og)^{\beta\alpha}\rb^\dagger\circ\d_x=\\
=&A^{\alpha\beta}\d_x+\d_x\circ\hOmega^1(\og)^{\alpha\beta}\circ\d_x+\d_x\circ\hOmega(\og)^{\alpha\beta}\circ\lb\frac{1}{2}-\mu_\beta\rb+\lb\frac{1}{2}-\mu_\alpha\rb\circ\hOmega(\og)^{\alpha\beta}\circ\d_x=\\
=&K^{\DR;\alpha\beta},
\end{align*}
which concludes the proof of the proposition.
\end{proof}

\end{document}